\documentclass[aps,pra,superscriptaddress,twocolumn,nofootinbib]{revtex4-2} 
\pdfoutput=1 
\usepackage{dsfont} 
\usepackage{comment} 
\usepackage[utf8]{inputenc} 
\usepackage[T1]{fontenc} 
\usepackage[british]{babel} 
\usepackage[dvipsnames,x11names]{xcolor} 
\usepackage{amsmath,amssymb,amsthm,bm,amsfonts,bbm} 
\usepackage{mathtools} 
\usepackage[colorlinks=true,citecolor=Green,linkcolor=Purple,urlcolor=Blue]{hyperref} 
\usepackage{physics}
\usepackage{amsfonts}

\newcommand{\id}{\mathds{1}}

\DeclareMathOperator{\sgn}{sgn}
\newtheorem{lemma}{Lemma}
\newtheorem{theorem}{Theorem}

\begin{document}
\title{Classical Cost of Transmitting a Qubit}

\author{Martin J. Renner}
    \email{martin.renner@univie.ac.at}
    \affiliation{University of Vienna, Faculty of Physics, Vienna Center for Quantum Science and Technology (VCQ), Boltzmanngasse 5, 1090 Vienna, Austria}
    \affiliation{Institute for Quantum Optics and Quantum Information (IQOQI), Austrian Academy of Sciences, Boltzmanngasse 3, 1090 Vienna, Austria}
    
\author{Armin Tavakoli}
    \affiliation{Institute for Quantum Optics and Quantum Information (IQOQI), Austrian Academy of Sciences, Boltzmanngasse 3, 1090 Vienna, Austria}
\affiliation{Atominstitut,  Technische  Universit{\"a}t  Wien, Stadionallee 2, 1020 Vienna, Austria}
\affiliation{Physics Department, Lund University, Box 118, 22100 Lund, Sweden}

\author{Marco Túlio Quintino}
\affiliation{Sorbonne Université, CNRS, LIP6, F-75005 Paris, France}
\affiliation{Institute for Quantum Optics and Quantum Information (IQOQI), Austrian Academy of Sciences, Boltzmanngasse 3, 1090 Vienna, Austria}
\affiliation{University of Vienna, Faculty of Physics, Vienna Center for Quantum Science and Technology (VCQ), Boltzmanngasse 5, 1090 Vienna, Austria}


\begin{abstract}
We consider general prepare-and-measure scenarios in which Alice can transmit qubit states to Bob, who can perform general measurements in the form of positive operator-valued measures (POVMs). We show that the statistics obtained in any such quantum protocol can be simulated by the purely classical means of shared randomness and two bits of communication. Furthermore, we prove that two bits of communication is the minimal cost of a perfect classical simulation. In addition, we apply our methods to Bell scenarios, which extends the well-known Toner and Bacon protocol. In particular, two bits of communication are enough to simulate all quantum correlations associated to arbitrary local POVMs applied to any entangled two-qubit state. 
\end{abstract}

\maketitle
\textit{Introduction.---} Quantum resources enable a sender and a receiver to break the limitations of  classical communication. When entanglement is available, classical \cite{Cleve1997,Buhrman2001, Leung2010, Tavakoli2021} as well as quantum communication \cite{Bennett2002, Piveteau2022} can be boosted beyond purely classical models. A seminal example is dense coding, in which two classical bits can be substituted for a single qubit and shared entanglement \cite{bennettdensecoding}. However, entanglement is not necessary for quantum advantages. Communicating an unassisted $d$-dimensional quantum system frequently outperforms the best conceivable protocols based on a classical $d$-dimensional system \cite{Buhrman10_review,Brassard2003, Vidick2011, Tavakoli2015, Navascues2015a}; even yielding advantages growing exponentially in $d$ \cite{hiddenmatching, raz1999}. Already in the simplest meaningful scenario, namely that in which the communication of a bit is substituted for a qubit, sizable advantages are obtained in important tasks like Random Access Coding \cite{WiesnerRAC1983, Ambainis2002, Ambainis2008}. These qubit advantages propel a variety of quantum information applications \cite{gallego10, Pawlowski2011, Li2011, Woodhead2015, Tavakoli2018}. 

It is natural to explore the fundamental limits of quantum over classical advantages. In order to do so, one has to investigate the amount of classical communication required to model the predictions of quantum theory. Previous works consider not only the scenario of sending quantum systems \cite{cerf2000, massar2001, Pati2000, tonerbacon2003, Methot2004}, but also simulating bipartite \cite{Maudlin1992, Brassard1999, gisingisin1999, Steiner2000, massar2001, Pati2000, cerf2000, tonerbacon2003, Methot2004, degorre2005, Degorre2007, Regev2010}, as well as multipartite entangled quantum systems \cite{Branciard2011, Branciard2012, Brassard2015, Brassard2019}. While such classical simulation of quantum theory is in general challenging, a breakthrough was made by Toner and Bacon \cite{tonerbacon2003}. Their protocol shows that any quantum prediction based on standard, projective, measurements on a qubit can be simulated by communicating only two classical bits. However, this does not account for the full power of quantum theory. More precisely, there exists qubit measurements that cannot be reduced to stochastic combinations of projective ones \cite{DAriano2005}. The most general measurements are known as positive operator-valued measures (POVMs). Physically, they correspond to the receiver interacting the message qubit with a locally prepared auxiliary qubit, and then performing a measurement on the joint system \cite{nielsen00}. Such POVMs are even indispensable for important tasks like unambiguous state discrimination \cite{IVANOVIC1987257, PERES198819} and hold a key role in many quantum information protocols (see e.g.~\cite{banaszek99, Renes2004, Vertesi2010, Bent2015, Acin2016, Curchod2017, Bae2019, Armin2020, Tavakoli2021b}). Importantly, they also give rise to correlations that cannot be modelled in any qubit experiment based on projective measurements \cite{Navascues2015, Tavakoli_2020, Mironowicz_2019, Steinberg_2021, Martinez2022}. 

This naturally raises the question of identifying the classical cost of simulating the most general predictions of quantum theory, based on POVMs. In the minimal qubit communication scenario, one may suspect that this cheap price of only two bits is due to the restriction to the, fundamentally binary, projective measurements. In contrast, when measurements are general POVMs, it is even unclear whether the classical simulation cost is finite. Notably, previous work has shown that there exists a classical simulation that requires 5.7 bits of communication on average \cite{cerf2000,Methot2004}. However, that protocol has a certain probability to fail in each round, leading to an unbounded amount of communication in the worst case.



In this work, we explicitly construct a classical protocol that simulates all qubit-based correlations in the prepare-and-measure scenario by using only two bits of communication. Thus, we find that the cost of a classical simulation remains the same when considering the most general class of measurements, although POVMs enable more general quantum correlations than projective measurements. Moreover, we show that two bits is the minimal classical simulation cost, i.e.~there exists no classical simulation that uses less communication than our protocol. This is shown through an explicit quantum protocol, based on qubit communication, that eludes simulation with a ternary classical message. Finally, we apply our methods to Bell nonlocality scenarios \cite{Brunner_2014}. We present novel protocols that simulate the statistics of local measurements on entangled qubit pairs.

\begin{figure}[ht]
    \centering
    \includegraphics[width=\columnwidth]{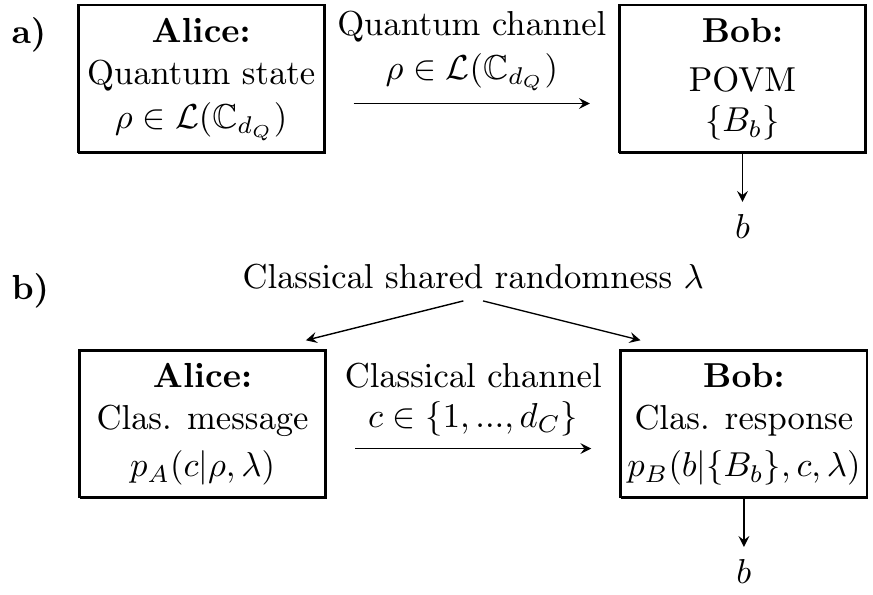}
    \caption{ \textbf{a)} Quantum PM scenario: Alice sends a $d_Q$-dimensional state to Bob who performs a POVM to obtain his outcome. \textbf{b)} Classical PM scenario for simulating the quantum PM scenario: The classical simulation is successful if, for every state and POVM, the probability that Bob outputs $b$ is the same as in the quantum protocol.} \label{fig:PM}
\end{figure}

\textit{The prepare-and-measure scenario.---} A quantum prepare-and-measure (PM) scenario (see Fig.~\ref{fig:PM}~a)) consists of two steps. First, Alice prepares an arbitrary quantum state of dimension $d_Q$ and sends it to Bob. The state is described by a positive semidefinite $d_Q\times d_Q$ complex matrix $\rho\in\mathcal{L}(\mathbb{C}_{d_Q})$, $\rho\geq 0$ with unit trace $\tr(\rho)=1$. Second, Bob receives the state and performs an arbitrary quantum measurement on it, obtaining an outcome $b$. General quantum measurements are described by a POVM, which is a set of operators $\{B_b\}$ that are positive semidefinite, $B_b\geq 0$ and sum to the identity, $\sum_b B_b=\mathds{1}$. In quantum theory, the probability of outcome $b$ when performing the POVM $\{B_b\}$ on the state $\rho$ is given by Born's rule, 
\begin{align}
	p_Q(b|\rho,\{B_b\})=\tr(\rho\, B_b)\, .\label{Born}
\end{align}

We are interested in constructing classical models for the PM scenario that simulate the predictions of quantum theory, i.e.~classical models that reproduce the probability distribution \eqref{Born}. In a classical simulation (see Fig.~\ref{fig:PM}~b)), Alice and Bob may share a random variable $\lambda$ subject to some probability function $\pi(\lambda)$. This allows them to correlate their classical communication strategies. Alice uses $\lambda$ and her knowledge of the quantum state $\rho$ to choose a classical message $c$ selected from a $d_C$-valued alphabet $\{1,\ldots,d_C\}$. Since the selection can be probabilistic, her actions are described by the conditional probability distribution $p_A(c|\rho,\lambda)$. When Bob receives the message, he uses $\lambda$ and his knowledge of the POVM $\{B_b\}$ to choose his outcome $b$. Again, this choice can be probabilistic and is therefore described by a conditional probability distribution $p_B(b|\{B_b\},c,\lambda)$. All together, the correlations obtained from the classical model become

\small
\begin{align}
p_C(b|\rho,\{B_b\})=
	 \int_\lambda   \text{d}\lambda \; \pi(\lambda) \sum_{c=1}^{d_C}  p_A(c|\rho,\lambda) p_B (b|\{B_b\},c,\lambda) \, . \label{eq:PMdef}
\end{align}
\normalsize

The simulation is successful if, for any choice of $d_Q$-dimensional states and POVMs, the quantum predictions $p_Q$ can be reproduced with a classical model using messages that attain at most $d_C$ different values. That is, if there exists a $d_C$ and suitable encodings $p_A$ and decodings $p_B$, such that
\begin{equation}
\forall \rho, \{B_b\}: \quad p_C(b|\rho,\{B_b\})=p_Q(b|\rho,\{B_b\}) \, .
\end{equation}
If this holds, we say that the classical model simulates quantum theory. In particular, we say that the classical simulation is minimal if no classical simulation is possible using a smaller message alphabet size $d_C$. Furthermore, we remark that for some PM scenarios, shared randomness may be charged as a non-free resource, leading to different results and problems \cite{massar2001,Galvao_2003,Ambainis2008, bowles14, bowles15, vicente17, Armin2020, Steinberg_2021, Krishna2022}. In fact, for the PM scenario we study here, it is known that an infinite amount of shared randomness is required in order to perform the task with finite classical communication \cite{massar2001}.

Our focus is on the most fundamental scenario, namely that based on qubits ($d_Q=2$). Notice that there exists a trivial classical simulation in which Alice sends the Bloch vector coordinates of her quantum state to Bob. After that, he can classically compute the Born rule and samples his outcome accordingly. However, sending the coordinates requires an infinite amount of communication ($d_C$ unbounded). Whether a classical simulation is possible with a finite value of $d_C$ is much less trivial. Notably, the simulation protocol of Toner and Bacon showed that if we additionally restrict the quantum measurements to be projective, i.e.~$B_b^2=B_b$, a classical simulation with $d_C=4$ (two bits) is possible \cite{tonerbacon2003}.

We also remark that here we consider a scenario where Bob does not know Alice’s state and Alice does not know Bob’s measurement beforehand. This scenario, where Alice and Bob can independently choose between different states and measurements, is even required to provide quantum over classical advantages in several tasks~\cite{WiesnerRAC1983, Ambainis2002, Ambainis2008, raz1999, hiddenmatching}. An interesting related scenario is the one where Bob’s measurement is known by Alice, or, equivalently, Bob has only a single choice of measurement. In that case, Frenkel and Weiner \cite{FrenkelWeiner2015} proved that, in the presence of shared randomness, a $d$-dimensional quantum system can always be perfectly simulated by a $d$-dimensional classical system. This powerful result inspired proposals such as the "No-Hypersignaling" principle~\cite{DallArno2017}, which is respected by quantum theory. In what follows, we find a minimal classical simulation for general qubit protocols.
		
\textit{Classical simulation protocol.---}\label{secprotocol} 
Qubit states $\rho$ can be represented as $\rho=\left(\openone+\vec{x}\cdot \vec{\sigma}\right)/2$, where $\vec{x}\in\mathbb{R}^3$ is a three-dimensional real vector such that  $|\vec{x}|\leq 1$, and $\vec{\sigma}=(\sigma_X,\sigma_Y,\sigma_Z)$ are the standard Pauli matrices. We may, without loss of generality, restrict ourselves to quantum protocols based on pure states. This corresponds to unit vectors $|\vec{x}|=1$. Since mixed states are convex combinations of pure states, every classical simulation protocol applicable to pure states can immediately be extended to apply also to mixed states. The classical randomness in the convex combination can simply be absorbed in the shared randomness of the simulation protocol.
%
%
Similarly, because every qubit POVM can be written as a coarse graining of rank-1 projectors~\cite{Barrett2002}, we may restrict ourselves to POVMs proportional to rank-1 projectors. Thus, we write Bob's measurements as $B_b=2p_b \ketbra{\vec{y}_b}$, where $p_b\geq0$, $\sum_b p_b=1$ and $\ketbra{\vec{y}_b}=\big(\mathds{1} + \vec{y}_b \cdot \vec{\sigma}\big)/2$ for some normalized vector $\vec{y}_b\in\mathbb{R}^3$. In Bloch notation we have
	\begin{equation}\label{qubitcorrelation}
   \tr(\rho B_b)=p_b(1+\vec{x}\cdot \vec{y}_b) .
\end{equation}	

We now present a classical simulation protocol in which Alice and Bob can perfectly simulate all qubit correlations at the cost of two bits of communication. To this end, it is handy to first define the Heaviside function, defined by $H(z)=1$ when $z\geq0$ and $H(z)=0$ when $z<0$, as well as the related function $\Theta(z):=z\cdot H(z)$. 
Consider now the following protocol.
\begin{enumerate}
	\item  Alice and Bob share two normalized vectors $\vec{\lambda}_1,\vec{\lambda}_2\in\mathbb{R}^3$, which are uniformly and independently distributed on the unit radius sphere $S_2$.
	\item Instead of sending a pure qubit $\rho={\big(\mathds{1} + \vec{x}\cdot\vec{\sigma}\big)/2}$, Alice prepares two bits via the formula ${c_1=H(\vec{x}\cdot \vec{\lambda}_1)}$ and $c_2=H(\vec{x}\cdot \vec{\lambda}_2)$ and sends them to Bob.
	\item Bob flips each vector $\vec{\lambda}_i$ when the corresponding bit $c_i$ is zero. More formally, he sets $\vec{\lambda}^{'}_i:=(-1)^{1+c_i}\vec{\lambda}_i$.
	\item Instead of performing a POVM with elements ${B_b=2p_b} \ketbra{\vec{y}_b}$, Bob picks one vector $\vec{y}_b$ from the set $\{\vec{y}_b\}$ according to the probabilities $\{p_b\}$. Then he sets $\vec{\lambda}:=\vec{\lambda}^{'}_1$ if $|\vec{\lambda}^{'}_1\cdot \vec{y}_b|\geq |\vec{\lambda}^{'}_2\cdot \vec{y}_b|$ and $\vec{\lambda}:=\vec{\lambda}^{'}_2$ otherwise. Finally, Bob outputs $b$ with probability
\begin{align} \label{eq:BobResponse}
    p_B(b|\{B_b\},\vec{\lambda})=\frac{p_b\  \Theta(\vec{y}_b\cdot \vec{\lambda})}{\sum_{j} \ p_j\  \Theta(\vec{y}_j\cdot \vec{\lambda})}  \, .
\end{align}
\end{enumerate}
The proof that the protocol perfectly reproduces the qubit correlations \eqref{qubitcorrelation} is given in Appendix~\ref{appendixa}. A sketch of the first three steps of the protocol is given in Fig.~\ref{fig2}.
\begin{figure}[ht]
    \centering
    \includegraphics[width=\columnwidth]{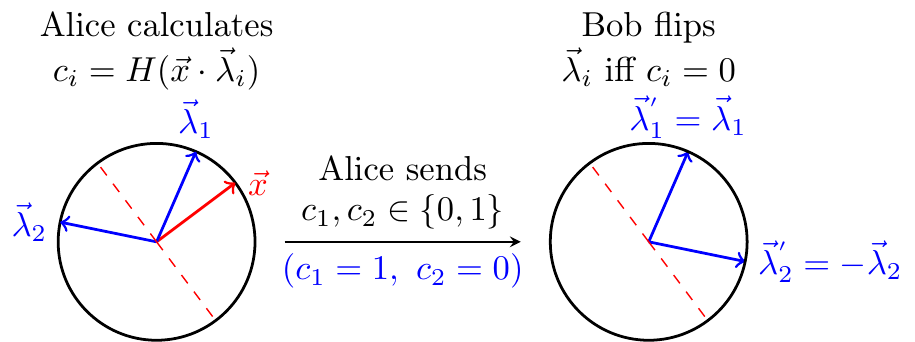}
    \caption{A two-dimensional illustration of the first three steps in the classical simulation protocol based on two bits. }
    \label{fig2}
\end{figure}
After the third step, the two vectors $\vec{\lambda}^{'}_1$ and $\vec{\lambda}^{'}_2$ are uniformly and independently distributed in the positive hemisphere defined by $\vec{x}$, i.e. their probability densities are $\rho(\vec{\lambda}^{'}_i)={H(\vec{x}\cdot \vec{\lambda}^{'}_i)/(2\pi)}$. As we show, this distribution is enough for Bob to classically reproduce the statistics of every POVM applied to the qubit state associated to $\vec{x}$.  Furthermore, in Appendix~\ref{appendixa} we also present a modified version of that protocol. There, Bob sends first one bit to Alice and then Alice sends one bit back to Bob.

\textit{Two bits are necessary for a classical simulation.---} We have shown that two classical bits are sufficient to simulate qubit correlations. We now prove that they are also necessary, i.e.~that the above classical simulation protocol is minimal. 

To this end, we show that there exists correlations in the qubit PM scenario that cannot be modelled in any classical protocol \eqref{eq:PMdef} that uses ternary messages ($d_C=3$). For this purpose, we consider PM scenarios with a fixed number of inputs for Alice and Bob. Alice selects her input from a set $x\in\{1,\ldots,I_A\}$ and prepares the qubit $\rho_x$. Bob selects his input from a set $y\in\{1,\ldots,I_B\}$ and performs the two-outcome projective measurement $\{B_{b|y}\}$ with outcomes labelled by $b\in\{1,2\}$. The qubit correlations are then given by $p_Q(b|x,y)=\tr(\rho_x B_{b|y})$. Notice that although Bob could perform POVMs, we are restricting ourselves to projective measurements. These turn out to be sufficient for the proof. 

It is key to recognise that the task of deciding whether a given $p(b|x,y)$ admits a classical simulation with a $d_C$-dimensional message alphabet can be solved by means of linear programming. From the duality theory of linear programming \cite{boyd_book}, we can obtain a classical dimension witness that certifies that the probabilities $p_Q(b|x,y)$ cannot be simulated by sending classical ternary messages. A classical dimension witness~\cite{raz1999,gallego10} is a linear inequality which is respected by \textit{all} classical models in the PM scenario for a given $d_C$. This can in general be written as
\begin{align}
	\sum_{b,x,y} \gamma(b|x,y) p_C(b|x,y) \leq C_d\, ,
\end{align}	
for some coefficients $\gamma(b|x,y)\in\mathbb{R}$. Here, $C_d$ is the classical bound. A violation of this inequality certifies that no classical model using $d_C$ symbols can simulate $p_Q(b|x,y)$. In Appendix~\ref{appendixlinprog}, we detail these linear programming methods. Inspired by the efficient method to find local bounds of Bell inequalities presented in Ref.~\cite{araujo20}, we provide a new and efficient algorithm to obtain the classical $d_C$-dimensional bound $\leq C_d$ for any given set of coefficients $\{\gamma(b,x,y)\}$. Also, drawing inspiration from  Ref.~\cite{bavaresco21}, we developed computational methods to convert the numerical solutions obtained from standard solvers to rigorous computer-assisted proofs which do not suffer from numerical precision issues due to floating point arithmetic.

In this way, we have obtained several examples of qubit states and measurements that generate quantum correlations $p_Q(b|x,y)$ that do not admit a classical model for $d_C=3$. An elegant example is obtained from considering $I_A=6$ states that form an octahedron on the Bloch sphere. They correspond to the eigenstates of the three Pauli operators $(\sigma_X, \sigma_Y, \sigma_Z)$. 
We let Bob perform $I_B=24$ different projective measurements. The Bloch vectors of these measurements are oriented such that they point to the vertices of a snub cube  \cite{snub}, which is an Archimedean solid, inscribed in the Bloch sphere. This may be viewed as a PM variant of Platonic Bell inequality violations \cite{Gisin2020}. Specifically, the 24 measurement directions are obtained as follows. Let $\tau$ be the one real root of the polynomial $x^3-x^2-x-1$, known as the Tribonacci constant. Take all even (odd) permutations of,  $(\pm 1, \pm 1/\tau, \pm \tau)$ and for each permutation, take only the four sign combinations that have an even (odd) number of  ``+''. This gives all vertices of the snub cube. Finally, do a global rotation by $60$ degrees in the XY-plane, i.e.~apply the unitary $U=\ketbra{0}{0}+e^{\frac{i\pi}{3}}\ketbra{1}{1}$ to all projectors. The linear programming methods reveal that the resulting $p_Q$ has no classical model for $d_C=3$.

In Appendix~\ref{appendixlinprog}, we discuss a heuristic aproach to find states and measurements leading to probabilities which do not admit a classical simulation for $d_C=3$. Fixing the above six preparations, the sparsest proof we have found uses eleven measurements that correspond to the solution of the Thomson problem \cite{Thomson}. All our computational code is openly available at the online repository~\cite{mtqGIT}.

Although no ternary message protocol is sufficient, it may still be that a classical simulation is possible by sending less than two bits on average. For example, Alice may restrict herself to send in some fraction of rounds only a trit, a bit or no communication at all. For the case of sometimes sending a bit or less, we show in Appendix~\ref{apponebit} that no classical simulation is possible. The reason is closely connected to the zero local weight of the singlet state, also known as the EPR2  decomposition~\cite{elitzur1992,barrett2006}. Our argument shows that, if one could simulate qubit correlations by sometimes sending only a bit or less, one could construct a protocol that simulates the singlet state without communication in these rounds. This would induce a local part for the singlet state, which contradicts the EPR2 decomposition.

\textit{Simulating  nonlocality.---} It is straightforward to adapt our classical protocol to simulate the statistics obtained from arbitrary local POVMs on any entangled qudit-qubit state. Indeed, all PM protocols can be adapted to Bell scenarios \cite{cerf2000}. For that, Alice chooses her measurement, an arbitrary POVM on a $d_Q$-dimensional quantum system. Then, she produces an output according to the marginal distribution of her POVM elements and, depending on her outcome, calculates the post-measurement state of Bob's qubit. Finally, she simply uses the classical protocol for the PM scenario to send that qubit state to Bob. Thus, our protocol immediately extends the best previously known one, due to Toner and Bacon \cite{tonerbacon2003}, to Bell scenarios involving POVMs. At the same time, we use the same amount of classical communication, in fact, two bits.

However, Toner and Bacon also show that only a single bit is necessary to simulate local projective measurements on a qubit pair in the singlet state $\ket{\Psi^-}=(\ket{01}-\ket{10})/\sqrt{2}$. We can also extend that result by constructing a novel one bit protocol. Here, Alice is restricted to projective measurements with outcomes $a=\pm 1$, but Bob can perform arbitrary POVMs.
\begin{enumerate}
	\item  Alice and Bob share two normalized vectors $\vec{\lambda}^{'}_1,\vec{\lambda}_2\in\mathbb{R}^3$, which are uniformly distributed on the unit radius sphere $S_2$. 
	\item Instead of performing a projective measurement with projectors $\ketbra{\pm\vec{x}}=\ (\mathds{1}\pm\vec{x}\cdot \vec{\sigma})/2$, Alice outputs $a=-\sgn(\vec{x}\cdot \vec{\lambda}^{'}_1)$ and sends the bit $c=\sgn(\vec{x}\cdot \vec{\lambda}^{'}_1)\cdot \sgn(\vec{x}\cdot \vec{\lambda}_2)$ to Bob. Here, $\sgn(z)=1$ when $z\geq 0$ and $\sgn(z)=-1$ when $z<0$.
	\item Bob flips the vector $\vec{\lambda}_2$ if and only if $c=-1$. More formally, he sets $\vec{\lambda}^{'}_2:=c\, \vec{\lambda}_2$.
	\item Same as "Step 4" in the original prepare-and-measure protocol.
\end{enumerate}
Since $\vec{\lambda}^{'}_1$ is uniformly distributed on $S_2$, we obtain the correct marginal probabilities $p(a)=1/2$ for Alice. Furthermore, when Alice outputs $a=+1$, $\vec{\lambda}^{'}_1$ and $\vec{\lambda}^{'}_2$ are distributed on $S_2$ according to $\rho(\vec{\lambda}^{'}_i)=H(-\vec{x}\cdot \vec{\lambda}^{'}_i)/(2\pi)$. This corresponds precisely to a classical description of Bob's post-measurement state $-\vec{x}$ (compare with the text below Fig.~\ref{fig2}). When Alice outputs $a=-1$, the two vectors are distributed according to $\rho(\vec{\lambda}^{'}_i)=H(+\vec{x}\cdot \vec{\lambda}^{'}_i)/(2\pi)$, which corresponds to the correct post-measurement state $+\vec{x}$. Therefore, Bob can apply the same response function ("Step 4") as in the original PM protocol, which immediately yields the correct quantum probabilities. 
Additionally, since singlet correlations have no local part \cite{elitzur1992,barrett2006}, one bit of communication is necessary in each round, ensuring the optimality of this protocol. Clearly, this protocol can be easily adapted to any maximally entangled qubit pair by rotating either Alice's or Bob's measurement basis.

\begin{table}
    \centering
    \begin{tabular}{c|c|c}
         Scenario& This work & Ref.~ \cite{tonerbacon2003}  \\\hline
         PM with qubit
         & 2 bits, POVMs & 2 bits, only Proj.\\
         Bell with 2 qubits & 2 bits, POVMs & 2 bits, only Proj.\\
         Bell with singlet & 1 bit, Proj.-POVM &  1 bit, Proj.-Proj. \\
    \end{tabular}
    \caption{Comparison between our protocol and the one by Toner and Bacon, previously the best protocol for these scenarios but restricted to only projective measurements (denoted as Proj. in this table) on Bob's side. Our protocols use the same resources, but Bob is allowed to perform POVMs.}
    \label{comparison}
\end{table}

\textit{Discussion.---} We have proven that two bits of communication are necessary and sufficient in order to classically simulate the most general predictions of quantum theory in a qubit prepare-and-measure scenario. Our results also have immediate implications for simulations of nonlocality in scenarios featuring POVMs. In this way, we generalised the well-known protocols of Toner and Bacon \cite{tonerbacon2003} from projective measurements to the most general qubit measurements (POVMs). Interestingly, this comes with no increase in the classical cost. See Table~\ref{comparison} for an overview.


A natural direction is to consider classical simulations for higher-dimensional quantum PM scenarios ($d_Q>2$), or scenarios involving entanglement. Notably, the latter can sometimes be isomorphic to the former \cite{Pauwels2022b}. Although this has received some attention \cite{Degorre2007, Regev2010, Brassard2019, Frenkel2022}, few general results are known. Most notably, it is still an open problem whether a qutrit ($d_Q=3$) PM scenario can be classically simulated with a finite message alphabet ($d_C<\infty$).

\begin{acknowledgments}
We thank \v{C}aslav Brukner, Valerio Scarani,  Peter Sidajaya, Isadora Veeren, Bai Chu Yu for fruitful discussions.
M.J.R. and M.T.Q. acknowledge financial support from the Austrian Science Fund (FWF) through BeyondC (F7103-N38), the Project No. I-2906, as well as support by the John Templeton Foundation through Grant 61466, The Quantum Information Structure of Spacetime (qiss.fr), the Foundational Questions Institute (FQXi) and the research platform TURIS. The opinions expressed in this publication are those of the authors and do not necessarily reflect the views of the John Templeton Foundation.
This project has received funding from the European Union’s Horizon 2020 research and innovation programme under the Marie Skłodowska-Curie grant agreement No 801110. It reflects only the authors' view, the EU Agency is not responsible for any use that may be made of the information it contains. ESQ has received funding from the Austrian Federal Ministry of Education, Science and Research (BMBWF).
A. T. is supported by the Wenner-Gren Foundation and by the Knut and Alice Wallenberg Foundation through the Wallenberg Center for Quantum Technology (WACQT).
\end{acknowledgments}


\nocite{apsrev42Control} 
\bibliographystyle{0_MTQ_apsrev4-2_corrected}
\bibliography{bib.bib}

\clearpage
\onecolumngrid
\appendix

\section{Proof of classical simulation protocol}\label{appendixa}
In this section, we prove that the classical  protocol based on two bits simulates qubit correlations in the PM scenario. First, we give a modified version of the protocol and show that both versions lead to the same statistics. 

\subsection{Modified version of the protocol}
The modified protocol, in which Bob sends one bit to Alice and afterwards Alice sends only one bit back to Bob, is:
	
\begin{enumerate}
	\item  Alice and Bob share two normalized vectors $\vec{\lambda}_1,\vec{\lambda}_2\in\mathbb{R}^3$, which are uniformly distributed on the unit radius sphere $S_2$.
	\item Instead of performing a POVM with elements ${B_b=2p_b} \ketbra{\vec{y}_b}$, Bob picks one vector $\vec{y}_b$ from the set $\{\vec{y}_b\}$ according to the probabilities $\{p_b\}$. Then he sets $k=1$ if $|\vec{\lambda}_1\cdot \vec{y}_b|\geq |\vec{\lambda}_2\cdot \vec{y}_b|$ and $k=2$ otherwise. Afterwards, he sends the bit $k$ to Alice.
	\item Given that $\rho={\big(\mathds{1} + \vec{x}\cdot\vec{\sigma}\big)}/2$ is the pure qubit state Alice wants to send, she only sends the bit $c_k=H(\vec{x}\cdot \vec{\lambda}_{k})$ to Bob.
	\item Bob flips the vector $\vec{\lambda}_{k}$ if the bit $c_k$ is zero. More formally, he sets $\vec{\lambda}:=(-1)^{1+c_k}\vec{\lambda}_{k}$.
	\item Finally, Bob outputs $b$ with probability
\begin{align} 
    p_B(b|\{B_b\},\vec{\lambda})=\frac{p_b\  \Theta(\vec{y}_b\cdot \vec{\lambda})}{\sum_{j} \ p_j\  \Theta(\vec{y}_j\cdot \vec{\lambda})}  \, . \label{EQA1}
\end{align}
\end{enumerate}

In the original version, Alice sends the two bits $c_1$ and $c_2$ that Bob needs to define the two vectors $\vec{\lambda}^{'}_i:=(-1)^{1+c_i}\vec{\lambda}_{i}$. Afterwards, Bob chooses one of the two vectors $\vec{\lambda}^{'}_i$ according to the test $|\vec{\lambda}^{'}_1\cdot \vec{y}_b|\geq |\vec{\lambda}^{'}_2\cdot \vec{y}_b|$ and proceeds only with the chosen vector $\vec{\lambda}:=\vec{\lambda}^{'}_k$. However, Bob's choice does only depend on the two vectors $\vec{\lambda}_1$ and $\vec{\lambda}_2$ but not on the bits $c_i$ he received from Alice since $|\vec{\lambda}^{'}_i\cdot \vec{y}_b|= |(-1)^{1+c_i}\vec{\lambda}_{i}\cdot \vec{y}_b|= |\vec{\lambda}_i\cdot \vec{y}_b|$. This observation allows us to modify the protocol. In the modified version, he makes his choice between $\vec{\lambda}_1$ and $\vec{\lambda}_2$ first. Afterwards, he informs Alice of his choice $\vec{\lambda}_k$ and Alice sends only the bit $c_k$ to Bob. This is enough for Bob to define the same $\vec{\lambda}:=(-1)^{1+c_k}\vec{\lambda}_{k}$.

\subsection{Proof of simulation protocol}
Before we present the proof, we show that the protocol is well-defined. More precisely, we can check that $p_B(b|\{B_b\},\vec{\lambda})$ are well-defined probabilities. In order to see this, note that $0\leq p_B(b|\{B_b\},\vec{\lambda})\leq 1$ follows from $\Theta(z)\geq 0$ (for every $z\in \mathbb{R}$) and $p_j\geq 0$ (for every $j$). Furthermore, we can check that
\begin{align}
    \sum_{b} p_B(b|\{B_b\},\vec{\lambda})=\frac{\sum_{b} \ p_b\  \Theta(\vec{y}_b\cdot \vec{\lambda})}{\sum_{j} \ p_j\  \Theta(\vec{y}_j\cdot \vec{\lambda})}=1 \, ,
\end{align}
to ensure normalisation.\\

\begin{theorem}
The above protocol reproduces the correct quantum probabilities. More precisely, for a given pure qubit state $\rho={ \big(\mathds{1} + \vec{x}\cdot\vec{\sigma}\big)}/2$ and POVM elements ${B_b=2p_b} \ketbra{\vec{y}_b}$, the total probability that Bob outputs $b$ is
\begin{align}
    p_C(b|\rho,\{B_b\})=p_b(1+\vec{x}\cdot \vec{y}_b)=\tr(\rho\, B_b)=p_Q(b|\rho,\{B_b\}) \, .
\end{align}
\end{theorem}
\begin{proof}
To check that the protocol outputs the correct probabilities, it is slightly more convenient to follow the modified version. We go through all the steps of the protocol and determine first the distribution of the vector $\vec{\lambda}_k$ after "Step 2", second the distribution of Bob's chosen vector $\vec{\lambda}$ after he performed "Step 4" and third we calculate the total probability that he outputs $b$ in "Step 5".\\

\textit{1. Distribution of $\vec{\lambda}_k$ after "Step 2":}\\
Alice and Bob share two vectors uniformly and independently distributed along the unit sphere $\vec{\lambda}_1,\vec{\lambda}_2 \in S_2$. Consider a round in which Bob has picked the POVM element that corresponds to the vector $\vec{y}_b$. Then he is choosing the vector $\vec{\lambda}_1$ if $|\vec{\lambda}_1\cdot \vec{y}_b|\geq |\vec{\lambda}_2\cdot \vec{y}_b|$ and $\vec{\lambda}_2$ otherwise. It follows from Degorre et al. "Theorem 6 (The “choice” method)" \cite{degorre2005} that the resulting distribution of the chosen vector $\vec{\lambda}_k$ is exactly:
\begin{align}
    \rho_b(\vec{\lambda}_k|\vec{y}_b)=\frac{1}{2\pi} |\vec{y}_b \cdot \vec{\lambda}_k| \, .
\end{align}
Since he is choosing $\vec{y}_b$ with probability $p_b$ the total distribution of the chosen vector $\vec{\lambda}_k$ is:
\begin{align}
    \rho(\vec{\lambda}_k|\{B_b\})=\sum_{b} p_b \ \rho_b(\vec{\lambda}_k|\vec{y}_b)=\frac{1}{2\pi} \sum_{b} p_b \ |\vec{y}_b \cdot \vec{\lambda}_k| \, .
\end{align}

\textit{2. Distribution of $\vec{\lambda}$ after "Step 4":}\\
Now Bob checks the received bit $c_k=H(\vec{x}\cdot \vec{\lambda}_k)$. He flips his chosen vector $\vec{\lambda}_k\rightarrow -\vec{\lambda}_k$ if and only if the received bit is zero. As a result, the distribution becomes:
\begin{align}
    \rho(\vec{\lambda}_k|\vec{x},\{B_b\})=2\cdot H(\vec{x}\cdot \vec{\lambda}_k)\cdot \rho(\vec{\lambda}_k|\{B_b\}) =\frac{H(\vec{x}\cdot \vec{\lambda}_k)}{\pi} \sum_{b} p_b \ |\vec{y}_b \cdot \vec{\lambda}_k| \, .
\end{align}
To see that this is true, note that if $H(\vec{x}\cdot \vec{\lambda}_k)=1$, Bob does not flip the vector $\vec{\lambda}_k$ and the distribution remains unchanged. If $H(\vec{x}\cdot \vec{\lambda}_k)=0$, Bob flips the vector. However, the distribution $\rho(\vec{\lambda}_k|\{B_b\})$ is point symmetric:
\begin{align}
    \rho(-\vec{\lambda}_k|\{B_b\})=\frac{1}{2\pi} \sum_{b} p_b \ |-\vec{y}_b \cdot \vec{\lambda}_k|=\frac{1}{2\pi} \sum_{b} p_b \ |\vec{y}_b \cdot \vec{\lambda}_k|=\rho(\vec{\lambda}_k|\{B_b\}) \, ,
\end{align}
from which the above expression follows. From here one, we can drop the index $k$ in $\vec{\lambda}_k$. We show below (Lemma~\ref{lemmaidentity}) that $\sum_{b} p_b \ |\vec{y}_b \cdot \vec{\lambda}|=2 \sum_{b} p_b \ \Theta(\vec{y}_b \cdot \vec{\lambda})$ and we use this to rewrite the distribution $\rho(\vec{\lambda}|\vec{x},\{B_b\})$ into:
\begin{align}
    \rho(\vec{\lambda}|\vec{x},\{B_b\})=\frac{2H(\vec{x}\cdot \vec{\lambda})}{\pi}\sum_{b} \ p_b\  \Theta(\vec{y}_b\cdot \vec{\lambda}) \, . \label{EQA8}
\end{align}

\textit{3. Total probability that Bob outputs $b$ in "Step 5":}\\
Finally, we are in a position to calculate the total probability that Bob outputs $b$ in "Step 5". Here, we use the expressions given in \eqref{EQA1} and \eqref{EQA8} to obtain:
\begin{align}
    p(b|\vec{x},\{B_b\})&=\int_{S_2}p_B(b|\{B_b\},\vec{\lambda})\cdot \rho(\vec{\lambda}|\vec{x},\{B_b\}) \  \mathrm{d}\vec{\lambda} =\frac{2p_b}{\pi} \int_{S_2} H(\vec{x}\cdot \vec{\lambda}) \cdot \  \Theta(\vec{y}_b\cdot \vec{\lambda})  \  \mathrm{d}\vec{\lambda} =p_b(1+\vec{x}\cdot \vec{y}_b) \, .
\end{align}
We evaluate the integral in Lemma~\ref{lemmaintegral} below. This equals exactly the required quantum statistics.
\end{proof}

\subsection{Evaluation of the integral}
\begin{lemma}\label{lemmaintegral}
Given two normalized vectors $\vec{x},\vec{y}\in \mathbb{R}^3$ on the unit sphere $S_2$, it holds that:
\begin{align}
    \frac{1}{\pi}\int_{S_2} H(\vec{x}\cdot \vec{\lambda})\cdot \ \Theta(\vec{y}\cdot \vec{\lambda})\  \mathrm{d}\vec{\lambda}=\frac{1}{2}(1+\vec{x}\cdot \vec{y}) \, ,
\end{align}
where $H(z)$ is the Heaviside function ($H(z)=1$ if $z\geq 0$ and $H(z)=0$ if $z< 0$) and $\Theta(z):=H(z)\cdot z$.
\end{lemma}
\begin{proof}
Note that both functions in the integral $H(\vec{x}\cdot \vec{\lambda})$ and $\Theta(\vec{y}\cdot \vec{\lambda})$ have support in only one half of the total sphere (the hemisphere centred around $\vec{x}$ and $\vec{y}$, respectively). For example, if $\vec{y}=-\vec{x}$ these two hemispheres are exactly opposite of each other and the integral becomes zero. For all other cases, we can observe that the value of the integral depends only on the angle between $\vec{x}$ and $\vec{y}$, because the whole expression is spherically symmetric. Therefore, it is enough to evaluate the integral for $\vec{x}=(0,1,0)^T$ and $\vec{y}=(-\sin{\beta},\cos{\beta},0)^T$, where we can choose without loss of generality $0\leq\beta\leq \pi$. Furthermore, we can use spherical coordinates for $\vec{\lambda}=(\sin{\theta}\cdot \cos{\phi}, \sin{\theta}\cdot \sin{\phi}, \cos{\theta})$ (note that $\vert\vec{\lambda}\vert =1$). With this choice of coordinates, the region in which both factors have non-zero support becomes exactly $\beta\leq \phi \leq \pi$ (at the same time, $\theta$ is unrestricted, $0 \leq \theta \leq \pi$). More precisely, $0\leq \phi \leq \pi$ is the support for $H(\vec{x}\cdot \vec{\lambda})$ and $\beta\leq \phi \leq \pi + \beta$ is the support for $\Theta(\vec{y}\cdot \vec{\lambda}$). In this way, the integral becomes:
\begin{align}
    \frac{1}{\pi}\int^{2\pi}_{0} \int^{\pi}_{0} H(\vec{x}\cdot \vec{\lambda})\cdot \ \Theta(\vec{y}\cdot \vec{\lambda}) \cdot \sin{\theta}\  \mathrm{d}\theta  \ \mathrm{d}\phi =\frac{1}{\pi}\int^{\pi}_{\beta} \int^{\pi}_{0} \ \sin{\phi} \cdot \sin^2{\theta}\  \mathrm{d}\theta  \ \mathrm{d}\phi =\frac{1}{2}(1+\cos{\beta})=\frac{1}{2}(1+\vec{x}\cdot \vec{y}) \, .
\end{align}
\end{proof}
It was recognized many times in the literature \cite{gisingisin1999, cerf2000, degorre2005} that the last expression exactly reproduces the statistics of measurements on qubits. However, in previous protocols, Alice was choosing a vector to create a distribution according to $\Theta(\vec{x}\cdot \vec{\lambda})$ and Bob outputs according to $H(\vec{y}\cdot \vec{\lambda})$. Here, we use the self-duality of quantum mechanics, which allows us to interchange the roles of states and measurements. In this sense, instead of Alice, Bob is choosing a vector to create a distribution like $\Theta(\vec{y}_b\cdot \vec{\lambda})$ and Alice contributes the term $H(\vec{x}\cdot \vec{\lambda})$ by telling Bob to flip that vector or not.

\subsection{Proof of a useful identity}
For the following Lemma, it is important to notice that for every POVM it holds that $\sum_{b} p_b\ \vec{y}_b=\vec{0}$. This follows from $\sum_b B_b=\mathds{1}$ with $B_b=2p_b \ketbra{\vec{y}_b}$, where $\sum_b p_b=1$ and $\ketbra{\vec{y}_b}=(\mathds{1}+\vec{y}_b\cdot \vec{\sigma})/2$:
\begin{align}
    \mathds{1}=\sum_b 2p_b \ketbra{\vec{y}_b}=
    \sum_b p_b (\mathds{1}+\vec{y}_b\cdot \vec{\sigma})=
    \sum_b p_b \ \mathds{1}+ \sum_b p_b\  \vec{y}_b\cdot \vec{\sigma}=\mathds{1}+ \sum_b p_b\  \vec{y}_b\cdot \vec{\sigma}&&\implies&&\sum_b p_b\ \vec{y}_b\cdot \vec{\sigma}=0  \, .
\end{align}
The last equation can only hold if $\sum_{b} p_b\ \vec{y}_b=\vec{0}$.

\begin{lemma}\label{lemmaidentity}
Given a set of vectors $\vec{y}_b\in S_2$ that satisfy $\sum_{b} p_b\ \vec{y}_b=\vec{0}$ and the function $\Theta(z)$, which is defined by $\Theta(z)=z$ if $z\geq 0$ and $\Theta(z)=0$ if $z<0$, it holds for every $\vec{\lambda}\in S_2$ that:
\begin{align}
    \sum_{b} p_b \ |\vec{y}_b \cdot \vec{\lambda}|=2 \sum_{b} p_b \ \Theta(\vec{y}_b \cdot \vec{\lambda}) \, .
\end{align}
\end{lemma}
\begin{proof}
First we prove that $\sum_{b} p_b\ \Theta(\vec{y}_b\cdot \vec{\lambda})=\sum_{b} p_b\ \Theta(-\vec{y}_b\cdot \vec{\lambda})$. Here, we use that $z=\Theta(z)-\Theta(-z)$ (for all $z\in \mathbb{R}$):
\begin{align}
    \vec{0}&=\sum_{b} p_b\ \vec{y}_b\ \implies \
    0=\vec{0}\cdot \vec{\lambda}=\sum_{b} p_b\ \vec{y}_b\cdot \vec{\lambda}=\sum_{b} p_b\ (\Theta(\vec{y}_b\cdot \vec{\lambda})-\Theta(-\vec{y}_b\cdot \vec{\lambda}))=\sum_{b} p_b\ \Theta(\vec{y}_b\cdot \vec{\lambda})-\sum_{b} p_b\ \Theta(-\vec{y}_b\cdot \vec{\lambda}) \, .
\end{align}
In the second step, we use this observation and $|z|=\Theta(z)+\Theta(-z)$ (for all $z\in \mathbb{R}$) to calculate:
\begin{align}
    \sum_{b} p_b \ |\vec{y}_b \cdot \vec{\lambda}|&=\sum_{b} p_b \ (\Theta(\vec{y}_b \cdot \vec{\lambda})+\Theta(-\vec{y}_b \cdot \vec{\lambda}))=\sum_{b} p_b \ \Theta(\vec{y}_b \cdot \vec{\lambda})+\sum_{b} p_b \ \Theta(-\vec{y}_b \cdot \vec{\lambda})=2 \sum_{b} p_b \ \Theta(\vec{y}_b \cdot \vec{\lambda}) \, .
\end{align}
\end{proof}

\section{No classical simulation with a one-bit part}\label{apponebit}
In this section, we show that every protocol that simulates a qubit in a PM scenario cannot have a part in which Alice communicates only a single bit to Bob. Interestingly, for our argument it is enough to consider only projective measurements. Hence, we can write the two projection operators for Bob as $\ketbra*{\pm\vec{y}}=(\mathds{1}\pm\vec{y}\cdot \vec{\sigma})/2$. As before, Alice can choose an arbitrary qubit state $\rho=(\mathds{1}+\vec{x}\cdot \vec{\sigma})/2$ that we simply denote with its Bloch vector $\vec{x}$. With that notation, in a classical protocol that simulates a qubit in a PM scenario, Bob has to output $b=\pm 1$ with probability:
\begin{align}
    p^{PM}(b|\vec{x},\vec{y})=\frac{1}{2}(1+b\ \vec{x}\cdot \vec{y}) \, .
\end{align}
We show that, given a protocol that simulates the qubit in the PM scenario with a non-zero one-bit part exists, it can be rewritten into a protocol that simulates the singlet with a non-zero local part. The latter is prohibited by the result of Elitzur, Popescu and Rohrlich \cite{elitzur1992} (see also Barrett et al.~\cite{barrett2006}) and our hypothesis follows by contradiction. To fix the notation, if Alice and Bob want to reproduce the statistics of local projective measurements on the singlet state $\ket{\Psi^-}=(\ket{01}-\ket{10})/\sqrt{2}$, they can choose measurement projectors $\ketbra*{\pm\vec{x}}=(\mathds{1}\pm\vec{x}\cdot \vec{\sigma})/2$ (for Alice) and $\ketbra*{\pm\vec{y}}=(\mathds{1}\pm\vec{y}\cdot \vec{\sigma})/2$ (for Bob). Then, the task becomes to output $a,b=\pm 1$ with probabilities:
\begin{align}
    p^{\Psi^-}(a,b|\vec{x},\vec{y})=\frac{1}{4}(1-(b\cdot a) \ \vec{x}\cdot \vec{y}) \, .
\end{align}
The similar form of these two expressions already suggests a connection between a protocol that simulates a qubit in a PM scenario and a protocol that simulates the singlet. We use the index "PM" and "$\Psi^-$" to distinguish theses two scenarios. We want to mention that the following statement also covers the scenario of no communication in some fraction of rounds. This is just a special case of a one-bit strategy in which Bob's response does not depend on the received message.

\begin{lemma}
Given a protocol that exactly simulates any qubit strategy in a prepare-and-measure scenario. The fraction of rounds in which Alice is communicating only a single bit to Bob has measure zero. More precisely, we can decompose such a protocol into:
\begin{align}
   p^{PM}(b|\vec{x},\vec{y})=\int_\lambda \text{d}\lambda \,  \pi(\lambda) \sum_{c=\pm 1} p_A(c|\vec{x},\lambda) p_B(b|\vec{y},c,\lambda)+ \int_{\tilde{\lambda}} \text{d}{\tilde{\lambda}} \,  \pi(\tilde{\lambda}) \sum_{m=1}^d p_A(m|\vec{x},\tilde{\lambda}) p_B(b|\vec{y},m,\tilde{\lambda}) \, ,
\end{align}
and it has to hold that $\int_\lambda \text{d} \lambda \, \pi(\lambda)=0$. Here, the first term are all the strategies that can be implemented with a single bit $c=\pm 1$ of communication and the second term contains all the strategies that require a longer message $m$ (with $d>2$).
\end{lemma}

\begin{proof}

To recapitulate, in those rounds, where Alice is allowed to send only a single bit to Bob, Alice's bit $c=\pm 1$ can depend only on her input $\vec{x}$ and the shared random variable $\lambda$ (denoted as $p_A(c|\vec{x}, \lambda)$) and Bob has to determine his output $b$ based on his input $\vec{y}$, the message $c$ and the shared random variable $\lambda$ (denoted as $p_B(b|\vec{y}, c, \lambda)$). The important observation is, given that Alice sends the bit $c=+1$, if she wants to transmit the state $\vec{x}$, she necessarily has to send the bit $c=-1$, if she wants to transmit the state $-\vec{x}$. To see that this must be true, consider that Bob chooses in that round the measurement basis $\vec{y}=\vec{x}$. In that situation, he necessarily has to discriminate between the two states $\vec{x}$ and $-\vec{x}$. This is not possible if Alice sends the same bit for both states $\vec{x}$ and $-\vec{x}$. Therefore, $p_A(c|\vec{x},\lambda)=p_A(-c|-\vec{x},\lambda)$.

On the other hand, given that Bob chooses the measurement basis $\vec{y}$ and wants to produce the output $b$, it might be that Alice has chosen the state $\vec{x}=\vec{y}$. If $c=+1$ is the message for the state $\vec{x}=\vec{y}$ in this round, it has to hold that $p_B(b=+1|\vec{y},c=+1,  \lambda)=1$ and since $c=-1$ is then necessarily the message for the state $-\vec{x}=-\vec{y}$ it has to hold that $p_B(b=-1| \vec{y},c=-1, \lambda)=1$. Analogously, it is also possible that $c=-1$ is the message for the state $\vec{x}=\vec{y}$ and $c=+1$ is the message for the state $-\vec{x}=-\vec{y}$, in which a similar argument leads to $p_B(b=+1| \vec{y},c=+1, \lambda)=p_B(b=-1| \vec{y},c=-1, \lambda)=0$. In any case, it has to hold that $p_B(b| \vec{y},c=+1, \lambda)=p_B(-b| \vec{y},c=-1, \lambda)$.

Now they can use a protocol that simulates a qubit in a PM scenario to simulate the singlet state \cite{cerf2000}. More precisely, Alice chooses her measurement basis $\vec{x}$ and tosses a balanced coin (heads and tails with probability $1/2$ each). If the coin shows heads, she outputs $a=+1$ and uses the PM protocol from above to send the state $-\vec{x}$ to Bob, whereas if the coin shows tails, she outputs $a=-1$ and uses the protocol to sends the state $+\vec{x}$ to Bob (be aware of the anti-correlation in the singlet state). This procedure simulates the singlet state since:
\begin{align}
   p^{\Psi^-}(a,b|\vec{x},\vec{y})&=\frac{1}{2}\cdot \delta_{a,+1}\cdot p^{PM}(b|-\vec{x},\vec{y})+\frac{1}{2}\cdot \delta_{a,-1}\cdot p^{PM}(b|+\vec{x},\vec{y})\\
   &=\frac{1}{4}\cdot \delta_{a,+1}\cdot (1-b\ \vec{x}\cdot \vec{y})+\frac{1}{4}\cdot \delta_{a,-1}\cdot (1+b\ \vec{x}\cdot \vec{y})\\
   &=\frac{1}{4}(1-(b\cdot a)\ \vec{x}\cdot \vec{y}) \, .
\end{align}
However, we can also write down the explicit protocol:
\begin{align}
   p^{\Psi^-}(a,b|\vec{x},\vec{y})=&\frac{1}{2}\cdot \delta_{a,+1}\cdot p^{PM}(b|-\vec{x},\vec{y})+\frac{1}{2}\cdot \delta_{a,-1}\cdot p^{PM}(b|+\vec{x},\vec{y})\\
\begin{split}
   =&\frac{1}{2}\cdot \delta_{a,+1}\cdot\left(\int_\lambda \text{d} \lambda \, \pi(\lambda) \sum_{c=\pm 1} p_A(c|-\vec{x},\lambda) p_B(b|\vec{y},c,\lambda)+\int_{\tilde{\lambda}} \text{d}\tilde{\lambda} \,  \pi(\tilde{\lambda}) \sum_{m=1}^d p_A(m|-\vec{x},\tilde{\lambda}) p_B(b|\vec{y},m,\tilde{\lambda})\right)\\
   &+\frac{1}{2}\cdot \delta_{a,-1}\cdot\left(\int_\lambda \text{d}\lambda \,  \pi(\lambda) \sum_{c=\pm 1} p_A(c|+\vec{x},\lambda) p_B(b|\vec{y},c,\lambda)+\int_{\tilde{\lambda}} \text{d}{\tilde{\lambda}} \,  \pi(\tilde{\lambda}) \sum_{m=1}^d p_A(m|+\vec{x},\tilde{\lambda}) p_B(b|\vec{y},m,\tilde{\lambda})\right) \, . \label{explicitprotocol}
   \end{split}
\end{align}
As before, the protocol to simulate the singlet is decomposed into a one-bit part and a part that requires more communication. The one bit part is the sum of all the terms that contain $\lambda$ (and not $\tilde{\lambda}$) as the shared variable. Together, they can be written as $\int_{\lambda} \text{d}\lambda \,  \pi (\lambda)\,  p^{\Psi^-}(a,b|\vec{x}, \vec{y}, \lambda)$ where:
\begin{align}
\begin{split}
    p^{\Psi^-}(a,b|\vec{x}, \vec{y}, \lambda):=&\frac{1}{2}\cdot \delta_{a,+1}\cdot ( p_A(c=+1|-\vec{x},\lambda) p_B(b|\vec{y},c=+1,\lambda)+p_A(c=-1|-\vec{x},\lambda) p_B(b|\vec{y},c=-1,\lambda))\\
    &+\frac{1}{2}\cdot \delta_{a,-1}\cdot (p_A(c=+1|+\vec{x},\lambda) p_B(b|\vec{y},c=+1,\lambda)+p_A(c=-1|+\vec{x},\lambda) p_B(b|\vec{y},c=-1,\lambda)) \, .
\end{split}
\end{align}
With the above relations $p_A(c|\vec{x},\lambda)=p_A(-c|-\vec{x},\lambda)$ and $p_B(b| \vec{y},c=+1, \lambda)=p_B(-b| \vec{y},c=-1, \lambda)$, we can rewrite this expression into:
\begin{align}
\begin{split}
    p^{\Psi^-}(a,b|\vec{x}, \vec{y}, \lambda)=&\frac{1}{2}\cdot \delta_{a,+1}\cdot p_A(c=-1|\vec{x},\lambda) p_B(b|\vec{y},c=+1,\lambda)+\frac{1}{2}\cdot \delta_{a,+1}\cdot p_A(c=+1|\vec{x},\lambda) p_B(-b|\vec{y},c=+1,\lambda)\\
    &+\frac{1}{2}\cdot \delta_{a,-1}\cdot p_A(c=+1|\vec{x},\lambda) p_B(b|\vec{y},c=+1,\lambda)+\frac{1}{2}\cdot \delta_{a,-1}\cdot p_A(c=-1|\vec{x},\lambda) p_B(-b|\vec{y},c=+1,\lambda) \, . \label{appequation}
\end{split}
\end{align}
The important observation is now that these correlations $p^{\Psi^-}(a,b|\vec{x}, \vec{y}, \lambda)$ can be realized with purely local strategies. More precisely, Alice and Bob share an additional random bit $r=\pm 1$ (with probability $1/2$ each). Given her measurement setting $\vec{x}$ and the shared random variable $\lambda$, Alice samples $c=\pm 1$ according to the probabilities $p_A(c|\vec{x}, \lambda)$ (as for the case of the message $c$ in the PM scenario). However, instead of sending the message $c$ to Bob, she outputs $a=-r\cdot c$. At the same time, Bob outputs $b=r\cdot b_{+}$ where $b_{+}$ is sampled according to the probabilities $p_B(b_{+}|\vec{y},c=+1,\lambda)$. If both follow that strategy and $r=+1$, they implement the behaviour
\begin{align}
    p^{\Psi^-}(a,b|\vec{x}, \vec{y}, \lambda, r=+1):=\left(\delta_{a,+1}\cdot p_A(c=-1|\vec{x},\lambda) +\delta_{a,-1} \cdot p_A(c=+1|\vec{x},\lambda)\right)\cdot p_B(b|\vec{y},c=+1,\lambda) \, .
\end{align}
On the other hand, if $r=-1$ they implement
\begin{align}
    p^{\Psi^-}(a,b|\vec{x}, \vec{y}, \lambda, r=-1):=\left(\delta_{a,+1}\cdot p_A(c=+1|\vec{x},\lambda) +\delta_{a,-1} \cdot p_A(c=-1|\vec{x},\lambda)\right)\cdot p_B(-b|\vec{y},c=+1,\lambda) \, .
\end{align}
It is easy to check that the weighted sum of these two expressions equals exactly the expression $p(a,b|\vec{x}, \vec{y}, \lambda)$ given in Eq.~\eqref{appequation}:
\begin{align}
    p^{\Psi^-}(a,b|\vec{x}, \vec{y}, \lambda)=\frac{1}{2}\cdot  p^{\Psi^-}(a,b|\vec{x}, \vec{y}, \lambda, r=+1)+\frac{1}{2}\cdot p^{\Psi^-}(a,b|\vec{x}, \vec{y}, \lambda, r=-1) \, .
\end{align}
Therefore, we have optimized the above protocol given in Eq.~\eqref{explicitprotocol}: Whenever Alice and Bob draw a $\lambda$ that corresponds to a one-bit part for the PM scenario, they can switch to the local strategy if they want to simulate the singlet. In the remaining rounds (according to shared randomness $\tilde{\lambda}$), where Alice was allowed to send more information, they do the same as in the case of the PM scenario: Alice outputs $-r$ and sends the message according to the state $r\vec{x}$. Bob outputs $b$ according to his message $m$ and his measurement basis $\vec{y}$.

Hence, given a simulation of the PM scenario with a non-zero one-bit part exists, we found a simulation of the singlet state with a non-zero local part. Since this is in contradiction with the result of Elitzur, Popescu and Rohrlich \cite{elitzur1992} (see also Barrett et al. \cite{barrett2006}), it proves our hypothesis.

\end{proof}

\section{Linear programming and classical PM scenarios} \label{appendixlinprog}

Let us consider a fixed PM scenario where Alice can prepare $I_A\in\mathbb{N}$ different inputs and Bob has $I_B\in\mathbb{N}$ different measurements with $O_B\in\mathbb{N}$ outcomes each.
The problem of deciding if a set of probabilities $\{p(b|x,y)\}$ can be obtained by Alice sending classical $d_C$-dimensional systems to Bob is phrased as:

\begin{align} \label{LP1}
	\text{given: }& \{p(b|x,y)\}, \; d_C  \\
	\text{find } \quad & \pi, p_A, p_B \\
\text{s.t.:\quad } & p(b|x,y) = \int_\lambda \text{d}\lambda \, \sum_{c=1}^{d_C} \, \pi(\lambda) \, p_A\big(c|x,\lambda\big)\, p_B \big( b|y,c,\lambda\big), \quad \forall b,x,y \label{eq:product}\\
&\pi(\lambda)\geq0, \; \forall \lambda \\
&\int_\lambda  \text{d}\lambda \pi(\lambda) =1 \\
&p_A\big(c|x,\lambda\big)\geq0, \; \forall c,x,\lambda\\
&\sum_{c=1}^{d_C}p_A\big(c|x,\lambda\big)=1, \; \forall x,\lambda,  \\
&p_B \big( b|y,c,\lambda\big)\geq0, \quad \forall b,x,y\\
&\sum_{b=1}^{O_B} p_B\big(b|y,c,\lambda\big)=1, \; \forall y,c,\lambda\, .
\end{align} 
Since in Eq.~\eqref{eq:product} we have a product of the optimisation variables, the above problem is not in a linear programming form. In order to rewrite it as a linear programming, we note that, similarly to Bell nonlocality~\cite{Brunner_2014}, the classical message $c$ sent by Alice may be chosen deterministically for a given $x$ and $\lambda$. This is true because the choice of the set of distributions $\{p_A(c|x,\lambda) \}$ form a polytope where the vertices are given by  deterministic distributions $D_A(c|x,\lambda)$, where $\lambda\in\{1,\ldots,d_C^{I_A}\}$. 
We then define $p_B'\big(b|y,c,\lambda\big):=\pi(\lambda)p_B\big(b|y,c,\lambda\big)$, a transformation which allows us to write the  problem described in Eqs.~\eqref{LP1} as,
\begin{align} \label{LP2}
	\text{given: }& \{p(b|x,y)\}, \; d_C, \; \{D_A(c|x,\lambda)\}  \\
	\text{find } \quad & \pi, p_B' \\
\text{s.t.:\quad } & p(b|x,y) = \sum_{\lambda=1}^{d_C^{I_A}}  \sum_{c=1}^{d_C} \,  D_A\big(c|x,\lambda\big)\, p_B' \big( b|y,c,\lambda\big), \quad \forall b,x,y \\
&p_B' \big( b|y,c,\lambda\big)\geq0, \quad \forall b,y,\lambda\\
&\sum_{b=1}^{O_B} p_B'\big(b|y,c,\lambda\big)=\pi(\lambda), \; \forall y,c,\lambda.
\end{align}
where $\sum_\lambda \pi(\lambda)=1$ follows from the fact that $\sum_b p(b|x,y)=1$. Note that now, all constraints are linear or positivity constraints, hence, the problem in Eqs.~\eqref{LP2} is a linear program.

We remark that, the set of distributions $\{p_B(b|c,y,\lambda \}$ also form a polytope where the vertices are given by  deterministic distributions $D_B(c|x,\lambda)$, where $\lambda\in\{1,\ldots,O_B^{I_B}\}$. Hence, one may construct a different linear program where both Alice and Bob have deterministic response functions. For practical reasons, this is often not a good choice, since it leads to a linear program with a big number of variables. In particular, the variable $\lambda$ would be allowed to take  $d_C^{I_A}O_B^{I_B d_C}$ different values as opposed to $d_C^{I_A}$. However, when considering a scenario where $d_c^{I_A}>O_B^{I_B d_C}$, it might be more efficient to set Bob as the part which performs deterministic strategies and to construct a different linear program by setting $p_A'\big(c|x,\lambda\big):=\pi(\lambda)p_A\big(c|x,\lambda\big)$.

\subsection{Primal formulation in terms of white noise robustness}

The linear program presented in Eqs.~\eqref{LP2} is a simple feasibility problem, since it only requires the existence of a feasible solution. We now adapt this feasibility problem to obtain a robustness optimisation problem. Instead of simply asking whether a set of probabilities $\{p(b|x,y)\}$ may be simulated by classical systems of dimension $d_C$, we look for the critical visibility parameter $\eta\in [0,1]$ such that the probabilities given by $\eta \; p(b|x,y) +(1-\eta)\frac{1}{O_B}$ admit a classical $d_C$-dimensional description. For that, we write the following linear program:
\begin{align} \label{LP3}
	\text{given: }& \{p(b|x,y)\}, \; \{D_A(.|x,\lambda)\}, \; d_C \\
	\max \quad &\eta  \\
\text{s.t.:\quad } & \eta \; p(b|x,y) +(1-\eta)\frac{1}{O_B}= \sum_{c=1}^{d_C}\sum_{\lambda}^{d_C^{I_A}} p_B'(b|y,c,\lambda) D_A(c|x,\lambda), \quad \forall b,y,x \quad \quad \\
	& p_B'(b|y,c,\lambda) \geq 0, \quad \forall b,y,c,\lambda \quad   \hspace*{23.5mm}\hspace*{36mm}\\
	& \sum_b p_B'(b|y,c,\lambda) = \; \pi(\lambda)\, .
\end{align}

\subsection{Classical dimension witness emerging from the dual problem }
We will now show how to obtain a classical dimension witness from the linear program presented in Eqs.~\eqref{LP3}. For the sake of concreteness, we will explicitly obtain the dual form from the Lagrangian method, see \cite{boyd_book} for an introduction. We start by setting the dual variables as:
\begin{align}
	\text{given: }& \{p(b|x,y)\}, \; \{D_A(c|x,\lambda)\}\; d_C \\
	\max \quad &\eta  \\
\text{s.t.:\quad } & \eta \; p(b|x,y) +(1-\eta)\frac{1}{O_B}= \sum_{c=1}^{d_C}\sum_{\lambda=1}^{d_C^{I_A}} p_B'(b|y,c,\lambda) D_A(c|x,\lambda), \quad \forall b,y,x \quad \quad \Big[\text{dual: }\; \gamma(b|x,y)\Big]\\
	& p_B'(b|y,c,\lambda) \geq 0, \quad \forall b,y,c,\lambda \quad   \hspace*{26.8mm}\hspace*{36mm}\Big[\text{dual: }\; \rho(b|y,c,\lambda)\Big]\\
	& \sum_b p_B'(b|y,c,\lambda) = \; \pi(\lambda),  \quad \forall y,c,\lambda \quad  \hspace*{30.55mm}\hspace*{24.5mm}\Big[\text{dual: }\; s(y,c,\lambda)\Big] \, .
\end{align}
	The Lagrangian is then given by
\begin{align}
	L = \eta &+ \sum_{b,x,y}\gamma(b|xy )\Big( \eta \; p(b|x,y)+ \frac{1}{O^B}-\frac{\eta}{O_B} - \sum_{c,\lambda}p_B'(b|y,c,\lambda) D_A(c|x,\lambda) \Big)\\
			  &+ \sum_{b,y,a,\lambda} p_B'(b|y,c,\lambda)\rho(b|y,c,\lambda) \\
			  &+ \sum_{y,c,\lambda} s(y,c,\lambda)\Big( \Big[\sum_b p_B'(b|y,c,\lambda)\Big] - \pi(\lambda) \Big)\, ;
\end{align}
	If we factorise the primal variables we have:
\begin{align}
	L &= \eta\Big(1 + \Big[\sum_{b,x,y} \gamma(b|x,y)p(b|x,y)\Big] - \Big[\sum_{b,x,y} \frac{\gamma(b|x,y)}{O_B}\Big] \Big)\\
&+\sum_{b,y,c,\lambda} p_B'(b|y,c,\lambda)\Big(\Big[-\sum_x\gamma(b|x,y)D_A(c|x,\lambda)\Big]+\rho(b|y,c,\lambda)+s(y,c,\lambda)\Big)\\
 &+\sum_{\lambda}\pi(\lambda) \Big(-\sum_{y,c} s(y,c,\lambda)\Big)\\
  &+\sum_{b,x,y} \frac{\gamma(b|x,y)}{O_B} \, .
\end{align}

This leads to the dual:
\begin{align}
	\text{given: }& \{p(b|x,y)\}, \; \{D_A(c|x,\lambda)\},\; d_C \\
	\min \quad & \sum_{b,x,y} \frac{\gamma(b|x,y)}{O_B}  \\
\text{s.t.:\quad } & \rho(b|y,c,\lambda)\geq0 \quad \forall b,y,c,\lambda \label{eq:rho1}\\
& 1 + \Big[\sum_{b,x,y} \gamma(b|x,y)p(b|x,y)\Big] = \sum_{b,x,y} \frac{\gamma(b|x,y)}{O_B}   \label{eq:gamma2} \\
& \rho(b|y,c,\lambda)=\Big[\sum_x\gamma(b|x,y)D_A(a|x,\lambda)\Big] - s(y,c,\lambda)\quad \forall b,y,c,\lambda \label{eq:rho2}\\
& \sum_{y,a}s(y,c,\lambda)=0 \quad \forall \lambda
\end{align}
and, we can also combine Eq.~\eqref{eq:rho1} with Eq.~\eqref{eq:rho2} to write:
\begin{align}
	\text{given: }& \{p(b|x,y), \; \{D_A(c|x,\lambda)\},\; d_C \\
	\min \quad & \sum_{b,x,y} \frac{\gamma(b|x,y)}{O_B} \label{eq:gamma1} \\
\text{s.t.:\quad } & \sum_x\gamma(b|x,y)D_A(c|x,\lambda)\geq s(y,c,\lambda) \quad \forall b,y,c,\lambda \\
& 1 + \Big[\sum_{b,x,y} \gamma(b|x,y)p(b|x,y)\Big] = \sum_{b,x,y} \frac{\gamma(b|x,y)}{O_B}     \\
& \sum_{y,c}s(y,c,\lambda)=0 \quad \forall \lambda 
\end{align} 
Additionally, in order to have a more explicit hyperplane formulation, we use Eq.~\eqref{eq:gamma1} and Eq.~\eqref{eq:gamma2} to write:
\begin{align} \label{LPdual}
	\text{given: }& \{p(b|x,y)\}, \; \{D_A(c|x,\lambda)\},\; d_C \\
	\min \quad & 1 + \sum_{b,x,y} \gamma(b|x,y)p(b|x,y)  \\
\text{s.t.:\quad } & \sum_x\gamma(b|x,y)D_A(c|x,\lambda)\geq s(y,c,\lambda) \quad \forall b,y,c,\lambda \\
&  \sum_{b,x,y} \frac{\gamma(b|x,y)}{O_B} =  1 + \sum_{b,x,y} \gamma(b|x,y)p(b|x,y)  \\
& \sum_{y,a}s(y,c,\lambda)=0 \quad \forall \lambda
\end{align} 

	We then see that the $\gamma(b|x,y)$ are the coefficients of the inequality which witness a non-classical $d_C$-dimensional behaviour $\{p(b|x,y)\}$. Additionally, since strong duality holds, the solution of the primal and dual coincides. Now, for behaviours $\{p(b|x,y)\}$ which are realisable with classical systems of dimension $d_C$, the visibility $\eta$ respects $\eta\geq1$, hence 
\begin{align}
	1 + \Big[\sum_{b,x,y} \gamma(b|x,y)p(b|x,y)\Big]\geq 1\, ,
\end{align}
and  $\sum_{b,x,y} \gamma(b|x,y)p(b|x,y)\geq0$, with $0$ being the bound of the inequality $\{\gamma(b|x,y)\}$ for $d_C$-dimensional systems. Also, for behaviours $\{p(b|x,y)\}$ which are realisable with classical systems of dimension $d_C$, the visibility $\eta=1$ is attainable, we have that the bound  $\sum_{b,x,y} \gamma(b|x,y)p(b|x,y)=0$ is attainable by classical systems of dimension $d_C$.

\subsection{Heuristic method to find quantum probabilities without a $d_C$-dimensional classical simulation}

A brute force method to generate quantum probabilities is simply to sample random states $\rho_x$ and measurements $\{B_{b|y}\}$ and then using linear programming to check weather $p(b|x,y)=\tr(\rho_x \, B_{b|y})$ may be simulated by $d_C$-dimensional classical systems. A more guided strategy may be to consider states and measurements that are  rather uniformly spread. For qubits, one may choose vectors rather uniformly spread in the Bloch sphere and then construct states and projective measurements for it.

In order to find an example of a set of qubit probabilities $p(b|x,y)=\tr(\rho_x \, B_{b|y})$ which makes use of $I_A=6$ states and $I_B=11$ projective measurements, we have chosen states and measurements corresponding to the Thomson problem \cite{Thomson}, a family of vectors on the sphere which is defined for any number of vectors $N\in\mathbb{N}$. In our online repository~\cite{mtqGIT} we provide an implementation for this heuristic method and also the exact qubit states and measurements in which $p(b|x,y)=\tr(\rho_x \, B_{b|y})$ cannot be simulated by classical trits.

\subsection{Efficient algorithm for obtaining the classical bound $C_d$} \label{sec:classical_bound}

We now present a novel and efficient algorithm for obtaining the classical bound $C_d$ for any given set of real numbers $\{\gamma(b|x,y)\}$. Our method is based on the scheme for finding the local bound of Bell inequalities presented in Ref.~\cite{araujo20}. We let $D_A$ and $D_B$ be the set of all deterministic strategies which can be performed by Alice and Bob. Hence, by convexity, we can write:
\begin{align}
C_d:= &\max_\lambda\Bigg[ 	\sum_{b,x,y,c} \gamma(b|x,y) D_A(c|x,\lambda)D_B(b|c,y,\lambda) \Bigg] \\
=&\max_\lambda\Bigg[  \sum_{b,y} D_B(b|c,y,\lambda)  \sum_{x,c} \gamma(b|x,y)  D_A(c|x,\lambda) \Bigg]\\
=&\max_\lambda\Bigg[ \sum_{b,y} \max_b\Big[  \sum_{x,c} \gamma(b|x,y)  D_A(c|x,\lambda)\Big] \Bigg]\, .\\
\end{align}	
We then see that we only need to generate Alice's deterministic strategies and to obtain the largest value of a vector, steps which can be done very efficiently.
\subsection{Computer-assisted proof}

In order to avoid numerical errors from floating point arithmetic, we show how to certify that a set of quantum probabilities given by $p(b|x,y)=\tr(\rho_x\, B_{b|y})$ cannot be simulated by classical systems of dimension $d_C$ by making use of only integers. The first step is to ensure that the probabilities $p(b|x,y)=\tr(\rho_x\, B_{b|y})$ are stored with only integers or fractions, for that we will ensure that the states $\rho_x$ and the measurements given by $B_{b|y}$ do not make use of floating point. We may adapt the Algorithm 1 of Ref.~\cite{bavaresco21} to obtain a quantum state, $\rho_\texttt{OK}$ which is described by fractions of integers, and it is close to $\rho_\text{float}$:

\paragraph*{\textbf{\emph{Algorithm 1:}}}
\begin{enumerate}
    \item 
    \texttt{Construct the non-floating-point matrix 
    $\rho_{\text{frac}}$ 
    by truncating the matrix
    $\rho_\text{float}$  }  
    \item    \texttt{Define the matrix} $\displaystyle{\rho:=\frac{\rho_\texttt{frac}+\left(\rho_\texttt{frac}\right)^\dagger}{2}}$
    \texttt{to obtain a self-adjoint matrix $\rho$}
    \item  \texttt{Find a coefficient $\eta$ such that $\rho':=\eta \rho + (1-\eta)\id$ is positive semidefinite} 
    \item  \texttt{Output the operator }
$\displaystyle{\rho_\texttt{OK}= \frac{\rho'}{\tr(\rho')}}$.
\end{enumerate}
Notice that checking whether a matrix with integers is positive semidefinite may be done efficiently by the
Sylvester’s criterion.

We now adapt Algorithm 2 of Ref.~\cite{bavaresco21} to transform any set of floating point POVM $\{B_{b,\texttt{float}}\}_{b=1}^{O_B}$ into a POVM described by fractions of integers.

\paragraph*{\textbf{\emph{Algorithm 2:}}}
\begin{enumerate}
    \item 
    \texttt{Construct the non-floating-point matrices 
    $B_{b,\text{frac}}$ 
    by truncating the matrices
    $B_{b,\text{float}}$ } 
    \item
    \texttt{For the outcomes $b\in\{1,\ldots,O_B-1\}$, define the matrix} $\displaystyle{B_b:=\frac{B_{b,\texttt{frac}}+\left(B_{b,\texttt{frac}}\right)^\dagger}{2}}$\texttt{. For $b=O_B$, define $B_{O_B}:=\id -\sum_{b=1}^{O_B-1}B_b$.}
    \item  \texttt{Find a coefficient $\eta$ such that the matrices $B'_b:={\eta B_b + (1-\eta)\id}$ are positive semidefinite for every $b\in\{1,\ldots, O_B\}$.} 
    \item  \texttt{Output the operator $B_{b,\texttt{OK}}= B'_b$.}
\end{enumerate}

Now, using a set of probabilities $\{p(b|x,y)\}$ which does not make use of floating point, we can then proceed as follows.
    
    \paragraph*{\textbf{\emph{Algorithm 3:}}}
\begin{enumerate}
    \item 
    \texttt{Solve the dual problem presented in Eqs.~\eqref{LPdual} by standard efficient floating point linear programming methods and obtain the inequality with coefficients $\gamma_{\text{float}}(b|x,y)$. }
    \item
    \texttt{Truncate  $\gamma_{\text{float}}(b|x,y)$ to obtain  $\gamma_{\text{frac}}(b|x,y)$.}
    \item  \texttt{Use the algorithm presented in Section~\ref{sec:classical_bound} to obtain $C_d$, the classical dimension bound for the witness given by $\gamma_{\text{float}}(b|x,y)$ .} 
    \item  \texttt{Verify that
    $\sum_{b,x,y}\gamma_{\text{float}}(b|x,y) p_B(b|x,y)>C_d$.}
\end{enumerate}

A Matlab implementation of all code presented in this section and used in this paper is openly available at the online repository~\cite{mtqGIT}.

\end{document}